\documentclass[12pt]{iopart}

\usepackage{graphicx}
\usepackage{amsthm}
\expandafter\let\csname equation*\endcsname\relax
\expandafter\let\csname endequation*\endcsname\relax
\usepackage{amsmath}
\usepackage{algorithm}
\usepackage{algpseudocode}
\usepackage{url}

\theoremstyle{thmstyleone}%
\newtheorem{theorem}{Theorem}
\newtheorem{corollary}{Corollary}

\begin{document}

\title[An Accurate and Efficient Analytic Model of Fidelity]{An Accurate and Efficient Analytic Model of Fidelity Under Depolarizing Noise Oriented to Large Scale Quantum System Design}

\author{Pau Escofet\textsuperscript{1}, Santiago Rodrigo\textsuperscript{1}, Artur Garcia-Sáez\textsuperscript{2,3} Eduard Alarcón\textsuperscript{1}, Sergi Abadal\textsuperscript{1}, Carmen G. Almudéver\textsuperscript{4}}

\address{
$^1$ Universitat Politècnica de Catalunya, Barcelona, Spain\\
$^2$ Bercelona Supercomputing Center, Barcelona, Spain\\
$^3$ Qilimanjaro Quantum Tech, Barcelona, Spain\\
$^4$ Universitat Politècnica de València, Valencia, Spain
}
\ead{pau.escofet@upc.edu}
\vspace{10pt}

\begin{abstract}
Fidelity is one of the most valuable and commonly used metrics for assessing the performance of quantum circuits on error-prone quantum processors. Several approaches have been proposed to estimate circuit fidelity without executing it on quantum hardware, but they often face limitations in scalability or accuracy. In this work, we present a comprehensive theoretical framework to predict the fidelity of quantum circuits under depolarizing noise. Building on theoretical results, we propose an efficient fidelity estimation algorithm based on device calibration data. The method is thoroughly validated through simulation and execution on real hardware, demonstrating improved accuracy compared to state-of-the-art alternatives, with enhancements in prediction $R^2$ ranging from 4.96\% to 213.54\%.. The proposed approach provides a scalable and practical tool for benchmarking quantum hardware, comparing quantum software techniques such as compilation methods, obtaining computation bounds for quantum systems, and guiding hardware design decisions, making it a critical resource for developing and evaluating quantum computing technologies.
\end{abstract}

\section{Introduction}
\label{sec:introduction}

Quantum errors are one of the most significant challenges in current quantum computers \cite{preskill_quantum_2018, Proctor2022}, limiting the reliable execution of quantum algorithms. Quantum systems suffer from various types of errors, including thermal errors \cite{khatri2020information, Chapeau_Blondeau_2022}, crosstalk errors \cite{ding_systematic_2020, Sarovar2020detectingcrosstalk}, readout errors \cite{Nachman2020}, and operational errors, which occur when executing quantum gates. Quantum error correction \cite{PhysRevA.52.R2493, error_correction, Litinski2019gameofsurfacecodes, 10.1145/3123939.3123949, Xu2024, PRXQuantum.2.040101} and quantum error mitigation \cite{RevModPhys.95.045005, Takagi2022, PRXQuantum.3.010345} techniques have been proposed to suppress or mitigate those errors, increasing the reliability of quantum computations and enabling practical use of quantum devices for tasks such as optimization \cite{Moll_2018}, cryptography \cite{shor_polynomial_1997}, and simulation \cite{feynman_1982_simulating}.

Assessing the performance of quantum processors is a fundamental step in advancing quantum computing \cite{Proctor2022}. Reliable performance evaluation enables the design of future quantum architectures, facilitates the comparison of quantum software tools such as compilation strategies, and supports informed decision-making at all stages of quantum system development. As quantum hardware scales, understanding the impact of noise and errors on computations becomes crucial for optimizing both hardware and software components for practical applications \cite{preskill_quantum_2018}.  

Several methods have been proposed to estimate the reliability of quantum computations \cite{4127220, 10.1145/3297858.3304007, 10361567, 9251243, Vadali2023, 10.1145/3508352.3561118}, offering insights into the expected performance of quantum algorithms. However, existing approaches typically suffer from one of three significant limitations. Some methods lack \emph{scalability}, restricting their applicability to small-scale quantum processors and making them impractical for analyzing larger, more complex systems \cite{4127220}. Others lack a strong \emph{theoretical foundation}, resulting in predictions that fail to reflect the behaviour of the quantum states accurately \cite{10.1145/3297858.3304007, 10361567}. Additionally, many approaches struggle to \emph{generalize across different hardware platforms}, requiring adaptation or retraining when used beyond their originally targeted systems \cite{9251243, Vadali2023, 10.1145/3508352.3561118}. These limitations hinder the ability to make precise performance assessments, which are essential for benchmarking hardware, refining error mitigation techniques, and guiding quantum system design.  

This work proposes a theoretical model that analytically characterizes how depolarizing noise affects quantum states when executing a circuit on a quantum processor. This model serves as the foundation for a fidelity estimation algorithm that predicts the final quantum state fidelity for a given quantum circuit and calibration data of the quantum processor.

The rest of the paper is structured as follows: Section \ref{sec:background} presents the necessary background for this work and reviews state-of-the-art fidelity estimation techniques. Section \ref{sec:theorems} studies the effects of depolarizing noise on quantum states and develops the theoretical framework underlying our fidelity estimation algorithm. Section \ref{sec:methods} details the experimental methodology, and Section \ref{sec:results} presents and discusses the results obtained. Section \ref{sec:discussion} explores potential use cases for the proposed model and its broader implications. Finally, Section \ref{sec:conclusions} summarizes the key findings and conclusions of this study.

\section{Background and Related Work}
\label{sec:background}

Quantum fidelity \cite{nielsen_chuang_2010, doi:10.1080/09500349414552171} is a measure of the similarity between two quantum states and serves as an indicator of either the reliability of a computation or the presence of excessive errors. It is calculated using the density matrix representation of the quantum states ($\rho$ and $\sigma$), as defined in Equation (\ref{eq:fidelity}).

\begin{equation}
    F(\rho, \sigma) := ||\sqrt{\rho} \sqrt{\sigma}||_1^2 = \left( tr \sqrt{\sqrt{\rho} \sigma \sqrt{\rho}} \right)^2
    \label{eq:fidelity}
\end{equation}

If at least one of the two states ($\rho$ or $\sigma$) is pure, the fidelity can be expressed as the trace of the product of the pure state's density matrix and the other state, whether pure or mixed (\textit{i.e.}, $F(\rho, \sigma) = tr(\sigma \rho)$). Fidelity exhibits properties such as \textit{unitary invariance} and \textit{multiplicativity under tensor product} (Equations (\ref{eq:unitary_invariance_fidelity}) and (\ref{eq:fidelity_tensor})) \cite{doi:10.1080/09500349414552171}, both of which are essential for the analysis presented in this work.

\begin{equation}
    F(\rho, \sigma) = F(U \rho U^\dagger, U \sigma U^\dagger)
    \label{eq:unitary_invariance_fidelity}
\end{equation}

\begin{equation}
    F(\rho_1 \otimes \rho_2, \sigma_1 \otimes \sigma_2) = F(\rho_1, \sigma_1) \cdot F(\rho_2, \sigma_2)
\label{eq:fidelity_tensor}
\end{equation}

Estimating the fidelity of quantum circuits without executing them on quantum hardware is crucial to advancing the design of quantum processors. A reliable fidelity estimation method enables researchers to compare quantum systems, optimize compilation strategies, and predict performance across different hardware platforms. As quantum circuits grow in complexity with increasing qubit numbers, the estimation method must stay computationally efficient to ensure scalability.

Various techniques have been proposed to estimate the quantum fidelity of a given quantum circuit when executed on a specific quantum processor. These approaches include simulation-based methods, analytical models, and machine-learning strategies.

When simulating quantum circuits \cite{10.1145/3310273.3323053, 10.5555/3135595.3135617, 9605307, 10044854} noise can be added to the quantum state, a process known as statistical fault injection \cite{4127220}. The fidelity of the computation can be assessed by comparing a simulated noisy state with its noiseless counterpart. However, the simulation of quantum states is computationally expensive, with complexity growing exponentially with the number of qubits, restricting its application to small-scale circuits \cite{Boixo2018}. The accuracy of the estimation will depend on the noise model used in the simulation, as well as on the precision of the simulation itself, which may introduce errors arising from approximations (\textit{e.g.}, limiting the dimensionality of tensors in a tensor network simulation \cite{biamonte2017tensor}).

Analytical approaches for fidelity estimation \cite{10.1145/3297858.3304007, 10361567} typically rely on gate error rates obtained from device calibration, usually through randomized benchmarking \cite{PhysRevA.77.012307, PhysRevA.85.042311}, without assuming any error model in particular. An example of such methods is the Probability of Successful Trial (PST) \cite{10.1145/3297858.3304007}, also known as Estimated Success Probability (ESP), a commonly used metric for estimating the fidelity of a quantum computation by multiplying the individual gate fidelities ($F_{g}$) and measurement fidelities ($F_{m}$) of the given circuit (Equation (\ref{eq:esp})). However, ESP does not account for how errors propagate through the system or specify which qubits are affected by the errors, obtaining a circuit fidelity that does not depend on the circuit's structure but just on the number of gates and their respective fidelities. 

\begin{equation}
    ESP = \prod_{i=1}^{N_{gates}} F_{g_i} \cdot \prod_{i=1}^{N_{meas}} F_{m_i}
    \label{eq:esp}
\end{equation}

The Quantum Vulnerability Analysis (QVA) method \cite{10361567}, like ESP, uses individual gate fidelities as its basis. However, it introduces a hyperparameter ($w \in [0,1]$) that quantifies the degree of cross-error introduced by a two-qubit gate, requiring fine-tuning for each specific system. Additionally, unlike ESP, QVA accounts for the fidelity of each two-qubit gate twice (once for each qubit involved), resulting in consistently lower fidelity estimates compared to ESP. 

Machine-learning-based techniques \cite{9251243, Vadali2023, 10.1145/3508352.3561118} have also been proposed to estimate circuit fidelity. These approaches often target specific platforms and require extensive training data and retraining for different devices. This dependence on platform-specific training limits their adaptability and scalability, making them less suitable for general-purpose fidelity estimation.

To better understand the strengths and limitations of existing fidelity estimation methods, we compare their predictions against experimental results from an actual quantum processor. Figure \ref{fig:sota_comparison} presents the fidelity predictions of statistical fault injection (using Qiskit's simulator \cite{javadiabhari2024quantumcomputingqiskit}), ESP, and QVA for quantum circuits ranging from 2 to 8 qubits, alongside the success rate of the same circuits executed on IBM Q Kyiv (Eagle r3 processor \cite{chow_2021_ibm}). The process from obtaining the circuits to computing the success rate after readout mitigation is explained in Section \ref{sec:execution}.

Figure \ref{fig:sota_comparison} shows how both Qiskit's simulation and ESP inaccurately predict the fidelity observed in real quantum processors, consistently overestimating it. Conversely, QVA's performance highly depends on its hyperparameter $w$, with fidelity predictions varying significantly, up to 0.8, between the extreme values of $w = 0$ and $w = 1$ (\emph{QVA max} and \emph{QVA min} in Figure \ref{fig:sota_comparison}). This high variability introduces the additional challenge of determining the optimal $w$, which the original work \cite{10361567} addressed by training a machine learning model for this purpose. In our analysis, we use $w = 0.5$ (depicted as red points in Figure \ref{fig:sota_comparison}) as it provides a balanced prediction (see box plot in Figure \ref{fig:sota_comparison}) while also reporting results for $w = 0$ and $w = 1$. In the original work \cite{10361567}, the authors identify $w=0$ as a near-optimal value for large circuits. Therefore, we report results for both $w=0.5$ and $w=0$ when assessing fidelity estimation performance to ensure a fair comparison.

\begin{figure}
    \centering
    \includegraphics[width=1\linewidth]{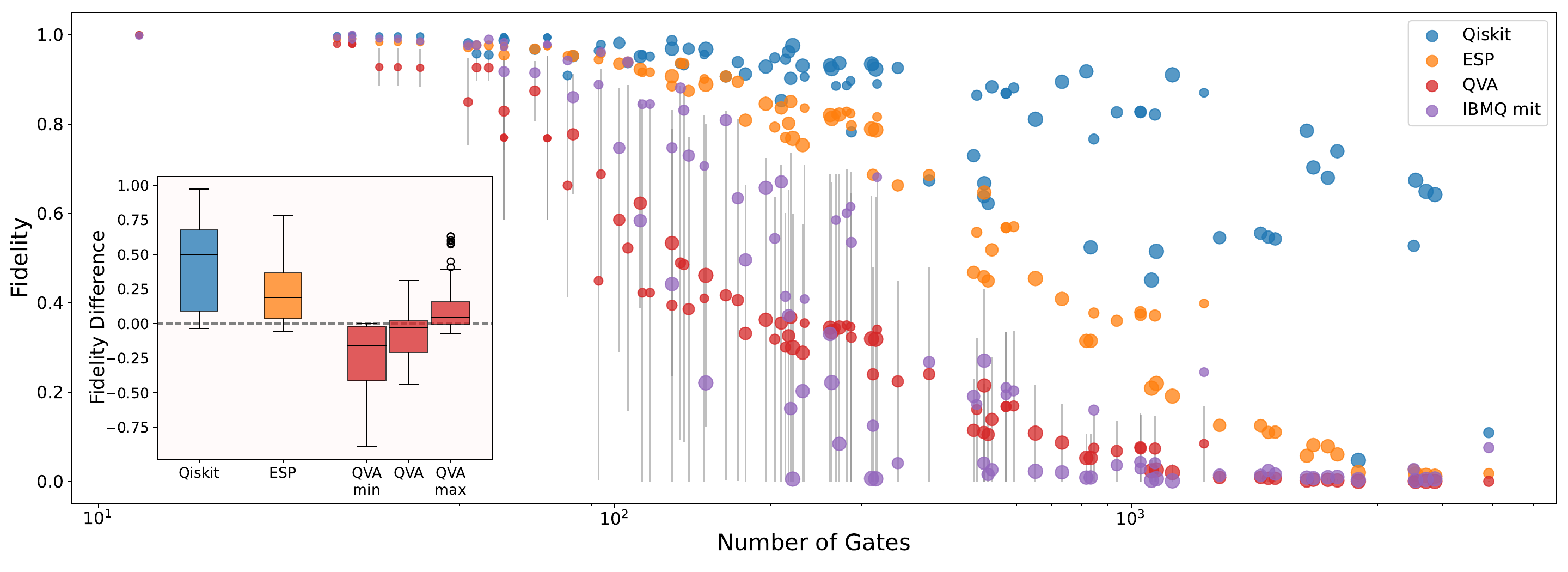}
    \caption{Predicted fidelity using various state-of-the-art prediction methods alongside the measured success rate from real hardware. For the QVA method, vertical lines indicate the range of solutions corresponding to different values of the $w$ hyperparameter. The size of the points is proportional to the number of qubits in the circuit (between 2 and 8). Boxplots display the distance between the predicted fidelity ($fid_{pred}$) and the measured success rate ($sr_{true}$), quantified as $fid_{pred} - sr_{true}$, providing insight into the accuracy and variability of each model's predictions.}
    \label{fig:sota_comparison}
\end{figure}

These observations highlight the need for a more accurate fidelity estimation model that, like QVA and ESP, is scalable in terms of computational cost, allowing the number of qubits and gates in the circuit to increase. However, minimizing the dependence on hyperparameters that can excessively influence predictions is essential, as this variability can compromise the reliability and consistency of the model's outcomes.

In this work, we develop a theoretical model for estimating the fidelity under depolarizing noise (Equation (\ref{eq:depolarizing_channel})), under the assumption that each quantum gate introduces depolarization to the involved qubits, a widely accepted perspective on operational errors \cite{PhysRevLett.119.180509, 10.1145/3674151, guinn2023codesignedsuperconductingarchitecturelattice, PRXQuantum.5.030326, PRXQuantum.4.020345, Proctor2022}. By examining how depolarizing noise impacts circuit fidelity, we lay the groundwork for a scalable fidelity estimation algorithm rooted in theoretical principles.

\begin{equation}
    \rho \rightarrow \mathcal{E}(\rho) = (1-p) \rho + \frac{p}{3} (X \rho X + Y \rho Y + Z \rho Z) = (1-p) \rho + \frac{p}{d} I_d
    \label{eq:depolarizing_channel}
\end{equation}

We primarily use the last expression in Equation (\ref{eq:depolarizing_channel}) to represent the depolarizing channel, where $d$ is the dimensionality of the quantum system being depolarized ($2^n$ for an \textit{n}-qubit system), $I_d/d$ is the maximally mixed state with the same dimensions ($d$) as $\rho$ (the quantum state described by its density matrix), and $p$ is the depolarization factor (or depolarization probability). This formulation provides a compact and general framework to model depolarization effects across systems of varying dimensions.

In the following two sections, we present our proposed fidelity estimation methodology. Section \ref{sec:theorems} develops the theoretical foundation of our approach, deriving a closed-form analytical model that characterizes how depolarizing noise affects quantum states throughout a circuit. Building on these theoretical insights, Section \ref{sec:methods} introduces a practical fidelity estimation algorithm that leverages the model to provide accurate predictions using only the quantum circuit structure and calibration data from the target device.

\section{Theoretical Modeling of the Fidelity Loss}
\label{sec:theorems}

In this section, we develop and prove the theorems that establish the foundation for a theoretical analysis of how the depolarizing noise channel influences the fidelity of quantum states.

To understand how depolarizing noise affects quantum circuits, we begin by exploring the evolution of a quantum state under the influence of repeated depolarizing channels. Theorem \ref{theo:n_depol} provides a formal expression for how a quantum state, described by its density matrix $\rho$, evolves after $n$ depolarizing channels have been applied, each introducing a depolarizing factor $p$.

\begin{theorem}[Consecutive depolarizing channels]
\label{theo:n_depol}
     After applying $n$ consecutive depolarizing channels with depolarizing factor $p$ to an initial pure quantum state $\rho$, the resulting quantum state $\mathcal{E}^n(\rho)$ can be expressed as:
    \begin{equation}
        \mathcal{E}^n (\rho) = (1-p)^n \rho + \frac{1-(1-p)^n}{d} I_d
    \end{equation}
\end{theorem}

\begin{proof}
    By induction.
    
    Let $P(n)$ be the statement that $\mathcal{E}^n(\rho) = (1-p)^n \rho + \frac{1-(1-p)^n}{d} I_d$.

    \textit{Base case:} For $n=1$ we have:
    \begin{equation}
        \mathcal{E}^1 (\rho)= (1-p)^1 \rho + \frac{1-(1-p)^1}{d} I_d = (1-p) \rho + \frac{p}{d} I_d
    \end{equation}
    
    This is true by the definition of depolarizing channel given in Equation(\ref{eq:depolarizing_channel}).

    \textit{Inductive hypothesis:} Assume $P(k)$ is correct for some positive integer $k$, which means $\mathcal{E}^k (\rho) = (1-p)^k \rho + \frac{1-(1-p)^k}{d} I_d$.

    \textit{Inductive step:} We now show that $P(k+1)$ is correct.

    \begin{align}
        \mathcal{E}^{k+1} (\rho) 
            &= \mathcal{E} \left( \mathcal{E}^k (\rho) \right) \\
            &= (1-p) \left( (1-p)^k \rho + \frac{1-(1-p)^k}{d} I_d \right) + \frac{p}{d} I_d \\
            &= (1-p)^{k+1} \rho + \frac{(1-p) - (1-p)^{k+1}}{d} I_d + \frac{p}{d} I_d \\
            &= (1-p)^{k+1} \rho + \frac{1- (1-p)^{k+1}}{d} I_d
    \end{align}
\end{proof}

Following this, Theorem \ref{theo:fidelity} quantifies the impact of these repeated depolarizations on the fidelity of the quantum state.

\begin{theorem}[Fidelity after $n$ depolarizations]
\label{theo:fidelity}
    The fidelity between a pure quantum state, $\rho$, and the same quantum state after $n$ depolarizing channels, $\mathcal{E}^n(\rho)$, is given by:
    \begin{equation}
        F(\rho, \mathcal{E}^n(\rho)) = (1 - p)^n + \frac{1-(1-p)^{n}}{d}
    \end{equation}
\end{theorem}

\begin{proof}
    \begin{align}
        F\left(\rho, \mathcal{E}^n (\rho) \right) 
            &= tr \left(\rho \cdot \mathcal{E}^n (\rho) \right) \\
            &= tr \left(\rho \cdot \left( (1-p)^n \rho + \frac{1-(1-p)^n}{d} I_d \right) \right) \\
            &= tr\left( (1-p)^n \rho^2 \right) + tr \left( \frac{1-(1-p)^n}{d} \rho \right) \\
            &= (1-p)^n tr\left( \rho^2 \right) + \frac{1-(1-p)^n}{d} tr \left( \rho \right) \\
            &= (1-p)^n + \frac{1-(1-p)^n}{d}
    \end{align}
\end{proof}

To further explore the impact of depolarization, Corollary \ref{theo:fidelity_loss} analyses the fidelity loss after a single depolarization, considering the effects of repeated depolarizations over time. This finding serves as a basis for quantifying the contribution of individual gate errors when estimating the fidelity of quantum circuits. 

\begin{corollary}[Fidelity loss]
\label{theo:fidelity_loss}
    The change in fidelity from a single depolarizing channel $\mathcal{E}^{n+1} (\rho) = \mathcal{E} ( \mathcal{E}^n (\rho))$ applied to an already depolarized state $\mathcal{E}^n (\rho)$ (previously pure state $\rho$) is:
    \begin{equation}
        F\left(\rho, \mathcal{E}^{n+1} (\rho)\right) = (1-p) F\left(\rho, \mathcal{E}^{n} (\rho)\right) + \frac{p}{d}
    \end{equation}
\end{corollary}

\begin{proof}
    \begin{align}
        F\left( \rho, \mathcal{E}^{n+1} (\rho) \right) 
            &= tr\left( \rho \cdot \mathcal{E}^{n+1} (\rho) \right) \\
            &= tr\left( \rho \cdot \left( (1-p) \mathcal{E}^n(\rho) + \frac{p}{d} I_d \right) (\rho) \right) \\
            &= (1-p) tr\left( \rho \cdot \mathcal{E}^n(\rho) \right) + \frac{p}{d} tr(\rho I_d) \\
            &= (1-p) F\left( \rho, \mathcal{E}^n(\rho) \right) + \frac{p}{d}
    \end{align}
\end{proof}

Theorem \ref{theo:two-qubit_depol} extends the previous results by focusing on a composite quantum state. We assume that the depolarizing error acting jointly on both subsystems can be decomposed into independent errors on each subsystem, leveraging the \textit{multiplicativity under tensor product} property of fidelity (Equation (\ref{eq:fidelity_tensor})). This enables us to track the fidelity of individual qubits in a multi-qubit system, which will be the foundation of the proposed fidelity estimation algorithm, described in Section \ref{sec:fidelity_alg}.

\begin{theorem}[Fidelity loss in a composite state]
\label{theo:two-qubit_depol}
    Let $\rho = \rho_A \otimes \rho_B$ be a pure, product state, and $\rho' = \mathcal{E}^i(\rho_A) \otimes \mathcal{E}^j(\rho_B)$ be the same product state after each component has undergone its respective depolarization process. The fidelity loss that a new depolarizing channel $\mathcal{E}(\rho') = \mathcal{E} \left( \mathcal{E}^i(\rho_A) \otimes \mathcal{E}^j(\rho_B) \right)$ applies to each previously depolarized quantum state is:

    \begin{eqnarray}
        F \left( \rho_A, \mathcal{E}(\rho')_A \right) = \sqrt{1-p} \cdot F \left(\rho_A, \mathcal{E}^i(\rho_A) \right) + \eta \\
        F \left( \rho_B, \mathcal{E}(\rho')_B \right) = \sqrt{1-p} \cdot F \left(\rho_B, \mathcal{E}^j(\rho_B) \right) + \eta
    \end{eqnarray}
    
    where $\mathcal{E}(\rho')_A$ and $\mathcal{E}(\rho')_B$ are the subsystems $A$ and $B$ in the composite state $\mathcal{E}(\rho')$, with dimensionality $d_{AB}$, and:
    \begin{align}
        \eta = & \frac{1}{2} \bigg( -\sqrt{1-p} \left( F\left(\rho_A, \mathcal{E}^i(\rho_A)\right) + F\left(\rho_B, \mathcal{E}^j(\rho_B)\right) \right) 
        \nonumber\\
        &+ \sqrt{(1-p) \left( F\left(\rho_A, \mathcal{E}^i(\rho_A)\right) + F\left(\rho_B, \mathcal{E}^j(\rho_B)\right) \right)^2 + \frac{4p}{d_{AB}} } \quad \bigg)
    \end{align}
\end{theorem}

\begin{proof}
    \begin{align}
        F(\rho, \mathcal{E}(\rho'))
            &= (1-p) \cdot F\left(\rho, \rho' \right) + \frac{p}{d_{AB}}\\
            &= (1-p) \cdot F\left(\rho_A, \mathcal{E}^i(\rho_A)\right) \cdot F\left(\rho_B, \mathcal{E}^j(\rho_B)\right) + \frac{p}{d_{AB}}\\
            &= \sqrt{1-p} \cdot F\left(\rho_A, \mathcal{E}^i(\rho_A)\right) \cdot \sqrt{1-p} \cdot F\left(\rho_B, \mathcal{E}^j(\rho_B)\right) + \frac{p}{d_{AB}} 
            \label{eq:eta1}\\ & \qquad\texttt{exists $\eta > 0$ s.t.} \nonumber\\
            &= \left(\sqrt{1-p} \cdot F\left(\rho_A, \mathcal{E}^i(\rho_A)\right) + \eta \right) \cdot \left(\sqrt{1-p} \cdot F\left(\rho_B, \mathcal{E}^j(\rho_B)\right) + \eta \right)
            \label{eq:eta2}\\
            &= F \left( \rho_A, \mathcal{E}(\rho')_A \right) \cdot F \left( \rho_B, \mathcal{E}(\rho')_B \right)
    \end{align}

    From Equations (\ref{eq:eta1}) and (\ref{eq:eta2}) we have:
    \begin{equation}
        \eta^2 + \eta \sqrt{1-p} \left( F\left(\rho_A, \mathcal{E}^i(\rho_A)\right) + F\left(\rho_B, \mathcal{E}^j(\rho_B)\right) \right) - \frac{p}{d_{AB}} = 0
    \end{equation}
    \begin{align}
        \eta =& \frac{1}{2} \bigg( -\sqrt{1-p} \left( F\left(\rho_A, \mathcal{E}^i(\rho_A)\right) + F\left(\rho_B, \mathcal{E}^j(\rho_B)\right) \right) 
        \nonumber\\
        &\pm \sqrt{(1-p) \left( F\left(\rho_A, \mathcal{E}^i(\rho_A)\right) + F\left(\rho_B, \mathcal{E}^j(\rho_B)\right) \right)^2 + \frac{4p}{d_{AB}} } \quad \bigg)
    \end{align}
    Since $\sqrt{1-p} \left( F\left(\rho_A, \mathcal{E}^i(\rho_A)\right) + F\left(\rho_B, \mathcal{E}^j(\rho_B)\right) \right) > 0$, and $\eta > 0$ we can discard the negative solution, resulting in
    \begin{align}
        \eta =& \frac{1}{2} \bigg( -\sqrt{1-p} \left( F\left(\rho_A, \mathcal{E}^i(\rho_A)\right) + F\left(\rho_B, \mathcal{E}^j(\rho_B)\right) \right) 
        \nonumber\\
        &+ \sqrt{(1-p) \left( F\left(\rho_A, \mathcal{E}^i(\rho_A)\right) + F\left(\rho_B, \mathcal{E}^j(\rho_B)\right) \right)^2 + \frac{4p}{d_{AB}} } \quad \bigg)
    \end{align}
    Which guarantees $\eta \geq 0$.
\end{proof}

Finally, Theorem \ref{theo:fidelity_entanglement} addresses the effect of depolarization on the fidelity of entangled qubits. This theorem explores the fidelity loss in a qubit entangled with another quantum state when a depolarizing channel is applied only to a subset of the entangled system.

\begin{theorem}[Fidelity loss in an entangled system]
\label{theo:fidelity_entanglement}
    Let $\rho_{AB}$ be a composite state in the $\mathcal{H}_A \otimes \mathcal{H}_B$ space. Suppose a depolarizing channel $\mathcal{E}$ is applied only to subsystem $A$. Let $\mathcal{E}_A (\rho_{AB}) = (1-p) \rho_{AB} + \frac{p}{d_A} \left( I_{d_A} \otimes tr_A(\rho_{AB}) \right)$ be the resulting state.

    The fidelity of subsystem $A$ between $\rho_{AB}$ and $\mathcal{E}_A (\rho_{AB})$ is bounded by:
    $$(1-p) \leq F\left(\rho_{AB}, \mathcal{E}_A (\rho_{AB})\right) \leq (1-p) + \frac{p}{d_A}$$
\end{theorem}

\begin{proof}
    \begin{align}
        F \left(\rho_{AB}, \mathcal{E}_A(\rho_{AB}) \right) 
            &= tr \left(\rho_{AB} \cdot \mathcal{E}_A(\rho_{AB})\right) \\
            &= tr \left( \rho_{AB} \cdot \left[ (1-p) \rho_{AB} + \frac{p}{d_A} \left( I_{d_A} \otimes tr_A(\rho_{AB}) \right) \right] \right) \\
            &= (1-p) tr(\rho_{AB}^2) + \frac{p}{d_A} tr\left( \rho_{AB} \cdot \left( I_{d_A} \otimes tr_A(\rho_{AB}) \right) \right) \\
            &= (1-p) + \frac{p}{d_A} tr\left( tr_B(\rho_{AB}) \right) tr\left( \left( tr_A(\rho_{AB}) \right)^2 \right) \\
            &= (1-p) + \frac{p}{d_A} tr\left( \left( tr_A(\rho_{AB}) \right)^2 \right)
    \end{align}

    If $\rho_{AB}$ is a pure product (\textit{i.e.}, separable) state, then $\rho_B = tr_A(\rho_{AB})$ is also a pure state and $tr\left( \left( tr_A(\rho_{AB}) \right)^2 \right) = 1$. If $\rho_{AB}$ is entangled (and thus not a product state), then $\rho_B$ is not pure, and $tr\left( \left( tr_A(\rho_{AB}) \right)^2 \right) < 1$.

    The exact value of $tr\left( \left( tr_A(\rho_{AB}) \right)^2 \right)$ lies between 0 and 1 and depends on the state $\rho_{AB}$, which can only be known through the simulation of the whole quantum state. From these values, we derive upper and lower bounds for the loss of fidelity:
    $$(1-p) \leq F(\rho_{AB}, \mathcal{E}_A (\rho_{AB})) \leq (1-p) + \frac{p}{d_A}$$
\end{proof}

Theorem \ref{theo:fidelity_entanglement} shows that the fidelity is upper bounded by the results found in Corollary \ref{theo:fidelity_loss}, obtaining a lower fidelity (higher fidelity loss) whenever the qubit being depolarized is entangled with other systems.

By presenting these theorems, a rigorous framework for understanding and quantifying the effects of depolarizing noise on quantum states and circuits is built, providing the foundation for the proposed fidelity estimation algorithm (see Section \ref{sec:fidelity_alg}), which, unlike ESP, QVA, or other analytical estimation techniques, is based on the theoretical impact of quantum errors on the quantum states, and not only on the device calibration data.

\subsection{Theoretical Model Validation}
\label{sec:theoretical_validation}
The algorithm developed using this theoretical framework avoids simulating the quantum circuit, making it inherently scalable to larger systems. This scalability comes at the cost of introducing a hyperparameter ($p_{ent} \in [0,1]$), which accounts for the entanglement level of the system to remain consistent with the findings in Theorem \ref{theo:fidelity_entanglement}. By adjusting $p_{ent}$, the model can effectively approximate the fidelity bounds for varying levels of entanglement without requiring computationally expensive simulations, obtaining lower and upper bounds on circuit fidelity for $p_{ent}=1$ and $p_{ent}=0$ respectively.

Figure \ref{fig:simulation_comparison} presents a comparison between the fidelities predicted by our theoretical model and those obtained from statistical fault injection simulations. The results span 93 quantum circuits, varying in size from 2 to 8 qubits and containing up to 5000 gates. Circuits are compiled to match the superconducting processor's restrictions used in latter experiments (see Section \ref{sec:execution}).

As expected, the fidelity predictions align closely with the simulated values, though slight discrepancies arise due to the theoretical upper and lower bounds employed in our model. To better illustrate these discrepancies, Figure \ref{fig:simulation_comparison} also shows the absolute difference between the predicted (with $p_{ent} = 0.5$) and simulated fidelities. The observed deviations remain within a reasonable range (for all data points, the maximum absolute fidelity difference remains under 0.07), confirming that our approach provides a robust approximation of circuit fidelity while being significantly more efficient than direct simulation.

\begin{figure}
    \centering
    \includegraphics[width=1\linewidth]{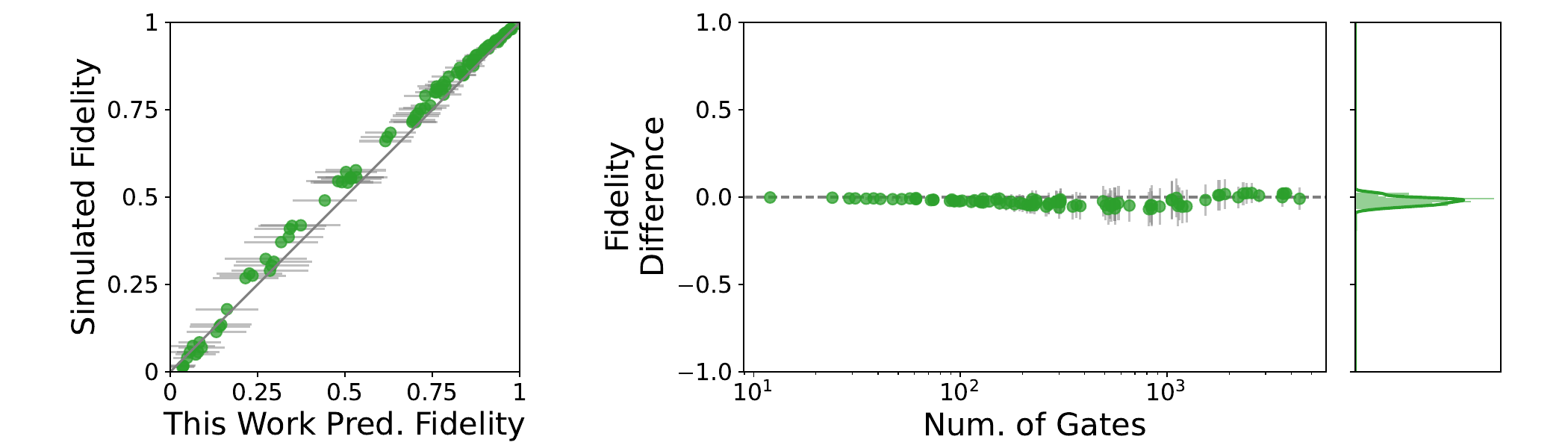}
    \caption{Comparison of predicted versus simulated fidelities with single-qubit gates depolarization of $p_1=10^{-3}$ and two-qubit gates depolarization of $p_2=5\cdot10^{-3}$. In the left plot, each point illustrates the predicted ($x$-axis) against the simulated fidelity ($y$-axis), with the diagonal line ($x=y$) indicating perfect predictions. Horizontal lines indicate the range of predicted fidelities arising from the variability of the hyperparameter $p_{ent}$. The right plot shows the difference in fidelity ($fid_{pred} - fid_{simul}$).}
    \label{fig:simulation_comparison}
\end{figure}

The results in Figure \ref{fig:simulation_comparison} also reveal that the fidelity range determined by $p_{ent}=0$ and $p_{ent}=1$ is relatively narrow, in contrast to QVA, where the hyperparameter $w$ can cause fidelity estimates to vary by as much as 0.8 (see Figure \ref{fig:sota_comparison}). This demonstrates that an intermediate value, such as $p_{ent}=0.5$, can reasonably estimate the fidelity. Consequently, the inability to precisely determine the level of entanglement does not significantly impact the accuracy of the fidelity estimation.

\section{Methods}
\label{sec:methods}

\subsection{Fidelity Estimation Algorithm}
\label{sec:fidelity_alg}

We propose the fidelity estimation algorithm based on the theoretical findings about fidelity loss under depolarizing noise presented in Section \ref{sec:theorems}. This algorithm predicts the expected fidelity of a quantum circuit based on its structure and the calibration data of the quantum processor on which the circuit is being executed, assuming that each gate introduces depolarizing noise. It is summarized in Algorithm \ref{alg:fidelity_estimation}.

\begin{algorithm}
\caption{Fidelity Estimation}
\label{alg:fidelity_estimation}
\begin{algorithmic}[0]
\Require Quantum circuit \texttt{circ}, $p_{ent} \in [0, 1]$ 
\Ensure $F_{circ}$ Circuit Fidelity
\State $\mathcal{Q} \gets \texttt{circ.qubits}$ \Comment{Set of qubits in the circuit}
\State $F_q \gets 1 \quad \forall q \in \mathcal{Q}$ \Comment{Initialize qubit-wise fidelity to 1}
\For {$\texttt{op} \in \texttt{circ}$}
    \State $d \gets 2 ^ {dim(\texttt{op.qubits})}$
    \State $p \gets \texttt{op.depolarization}$
    \If{$d == 2$} \Comment{Single-qubit gate}
        \State $q_i \gets \texttt{op.qubits}$
        \State $F_{q_i} \gets (1-p) \cdot F_{q_i} + (1-p_{ent}) \cdot \frac{p}{d}$ 
    \ElsIf{$d == 4$} \Comment{Two-qubit gate}
        \State $q_i, q_j \gets \texttt{op.qubits}$
        \State $\eta \gets \frac{1}{2} \left( \sqrt{ (1-p) \left( F_{q_i} + F_{q_j} \right)^2 + p} - \sqrt{1-p} \left( F_{q_i} + F_{q_j} \right)\right)$
        \State $F_{q_i} \gets \sqrt{1-p} \cdot F_{q_i} + (1-p_{ent}) \cdot \eta$
        \State $F_{q_j} \gets \sqrt{1-p} \cdot F_{q_j}+ (1-p_{ent}) \cdot \eta$
    \EndIf
\EndFor
\State \textbf{return} $\prod_{q \in \mathcal{Q}} F_q$
\end{algorithmic}
\end{algorithm}

The calibration data provides gate fidelities for all single- and two-qubit gates ($d=2$ and $d=4$ respectively) of a given quantum processor. From this, we derive the depolarization parameter $p$ for each single- and two-qubit gate, ensuring that the gate fidelity $F_g$ from calibration aligns with the theoretical fidelity loss due to a depolarization factor $p$ (obtained in Corollary \ref{theo:fidelity_loss}), which is given by:

\begin{equation}
 p = \frac{d (F_g - 1)}{1-d}
\end{equation}

By leveraging the unitary invariance of both the depolarizing channel and quantum fidelity, all gates in the circuit are effectively replaced with fidelity-equivalent depolarizing channels. Each qubit's independent fidelity is initialized to 1 and iteratively reduced based on the depolarizing channels applied to each quantum state. The global fidelity is determined by the \textit{multiplicativity under tensor product} property of fidelity. Specifically, the circuit's overall fidelity ($F_{circ}$) is obtained as the product of the individual qubit-wise fidelities ($F_q$):

\begin{equation}
    F_{circ} = \prod_{q \in \mathcal{Q}} F_q
\end{equation}

As proved in Theorem \ref{theo:fidelity_entanglement}, the fidelity loss induced by a depolarizing channel is influenced by the level of entanglement of the involved qubits. Accounting for this would require simulating all quantum states and tracking how each gate impacts the states throughout the circuit (\textit{i.e.}, to simulate the whole computation classically). This approach, however, is computationally prohibitive for systems with more than a few dozen qubits \cite{Boixo2018}. To address this, we introduce a hyperparameter $p_{ent}$, which defines the level of entanglement in the quantum states. No entanglement is assumed when $p_{ent}=0$, minimizing the fidelity loss and yielding an upper bound on the circuit's fidelity. Conversely, setting $p_{ent}=1$ assumes maximum entanglement, resulting in a more significant fidelity loss for each gate and providing a lower bound on the circuit's fidelity.

It is essential to incorporate the effects of coherence errors \cite{nielsen_chuang_2010, khatri2020information, Chapeau_Blondeau_2022} into the fidelity estimation model to validate the approach presented in this work against real quantum hardware, as these errors are a significant source of decoherence in current quantum hardware. The exact decoherence induced by relaxation ($T_1$) and dephasing ($T_2$) depends on the particular quantum state (\textit{i.e.}, they are not unitary-invariant), which can only be fully assessed by simulating the whole quantum circuit. Instead, we assume a worst-case scenario where qubits undergo maximal decoherence, thus providing an upper bound on fidelity loss.

The resulting decoherence upper bound is incorporated into Algorithm \ref{alg:fidelity_estimation} by scaling the qubit-wise fidelity after the execution of each layer of the circuit (\textit{i.e.}, a set of gates that can be applied simultaneously on independent qubits) to account for coherence errors. Specifically, the relaxation time $T_1$ and dephasing time $T_2$ for each qubit are obtained from the device calibration data, and the execution time for each layer is estimated based on the duration of the gates that can be executed in parallel. After each layer, the fidelity of each qubit (even the ones that remained idling during that layer) is scaled as:

\begin{equation}
\label{eq:thermal_depol_eq}
    \forall{q \in \mathcal{Q}} : F_q = F_q \cdot e^{-t_{\text{layer}}/T_1^q} \cdot \left( \frac{1}{2} e^{-t_{\text{layer}}/T_2^q} + 0.5 \right)
\end{equation}

This equation accounts for both relaxation and dephasing effects on fidelity, consistent with how IBMQ calibrates coherence times, modeling $T_1$ and $T_2$ using $A \cdot e^{-t/T} + B$. Resulting in $A \simeq 1$ and $B \simeq 0$ for $T_1$, and $A \simeq B \simeq 0.5$ for $T_2$ \cite{t1_ibm, t2_ibm}.

\subsection{Quantum Processor Execution}
\label{sec:execution}

This subsection details the experimental setup and workflow for evaluating circuit success rates and comparing fidelity estimation methods. Figure \ref{fig:workflow} provides an overview of the entire workflow.

We selected quantum circuits ranging from 2 to 8 qubits, obtained from the MQT Bench \cite{quetschlich2023mqtbench}. Each circuit $\mathcal{C}$ was concatenated with its inverse $\mathcal{C}^{-1}$, ensuring that the final quantum state is equal to the initial one, the all-zero state $|0\rangle ^{\otimes q}$ in our case, with a corresponding ideal measurement outcome of '0'$\times q$.

Before running the circuits on the actual quantum device, they went through the compilation process, which decomposed gates into native gates supported by the device \cite{10181370, 10.1145/3445814.3446718}, mapped logical qubits to physical qubits \cite{qubit_allocation}, and added routing operations \cite{li_2019_tackling, 10313754} to meet hardware connectivity constraints \cite{devoret2004superconductingqubitsshortreview}. We used Qiskit \cite{javadiabhari2024quantumcomputingqiskit} for all stages of the compilation process, ensuring consistency across all the evaluated circuits. The result is a compiled circuit optimized for execution on the chosen quantum processor.

\begin{figure}
    \centering
    \includegraphics[width=1\linewidth]{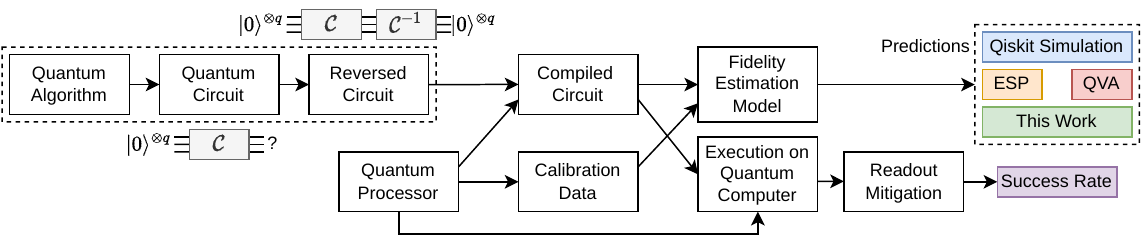}
    \caption{Workflow for obtaining fidelity predictions from the four estimation methods (Qiskit, ESP, QVA, proposed model, and proposed $+T_{1,2}$) and determining the success rate from real-device executions. The process includes circuit selection, compilation to match device-specific constraints, fidelity estimation based on noise models and calibration data, execution on a real quantum processor, and measurement error mitigation to compute the success rate. This workflow ensures a consistent comparison of predictions with hardware performance.}
    \label{fig:workflow}
\end{figure}

For each estimation method, fidelity predictions were made using the compiled circuit and calibration data from the quantum processor at the execution time. For Qiskit simulations, noise models were directly obtained from the Qiskit library \cite{javadiabhari2024quantumcomputingqiskit}, tailored to the target device, and simulations were conducted using the \texttt{density\_matrix} method to account for mixed-state evolution under noise.

Experiments were conducted on an IBM Eagle r3 processor \cite{chow_2021_ibm} with 127 qubits. The compiler selected the most reliable subset of physical qubits for each circuit. After execution, qubits were measured, and readout errors were mitigated using calibration data and assuming uncorrelated measurement errors between the qubits. The success rate is calculated as the ratio of correctly measured outcomes ('0'$\times q$) over the total number of executions (2048 per circuit).

This setup establishes a robust framework for evaluating the performance of the real quantum device against fidelity predictions from various estimation methods, providing valuable insights into each model's accuracy and dependability.

\section{Results}
\label{sec:results}

Having validated the model through simulation in Section \ref{sec:theoretical_validation}, we now assess its performance against real quantum hardware.

The fidelity for each compiled circuit in the benchmark is estimated using the proposed model alongside other methods, including Qiskit simulation (using noise models associated with the used processor), ESP, and QVA. We compare these predictions to the success rates obtained from executions on quantum hardware after readout error mitigation. Additionally, a version of each model incorporating coherence errors ($T_1$ and $T_2$), as detailed in Section \ref{sec:fidelity_alg}, is also considered to provide a more comprehensive assessment of actual executions. Since quantum hardware is subject to additional sources of error beyond those captured by the considered models, achieving an accurate prediction by considering only gate errors is impossible, making the inclusion of $T_1$ and $T_2$ essential for a robust validation against quantum hardware.

Figure \ref{fig:correlation} illustrates the measured success rate in the quantum device (IBM Q Eagle processor) after readout error mitigation (on the $y$-axis) versus the predicted fidelity (on the $x$-axis) for each estimation method. In this plot, a perfect prediction would place the data point along the diagonal. For the QVA and proposed models, horizontal lines indicate the influence of the hyperparameters in each algorithm ($w$ for QVA and $p_{ent}$ for the proposed model), with the reported point representing the average value between the extremes of these parameters. Most gate-error models tend to estimate a higher fidelity than the actual observation on the real device, which is expected since the models do not account for all types of errors present in actual hardware (\textit{e.g.}, coherence or crosstalk errors). This effect is evident when incorporating coherence errors into the estimation models, as it demonstrates a notable improvement in aligning predictions with the observed success rates for all models except for QVA, where it leads to a worse prediction, estimating a lower fidelity than the observed one.

\begin{figure}
    \centering
    \includegraphics[width=1\linewidth]{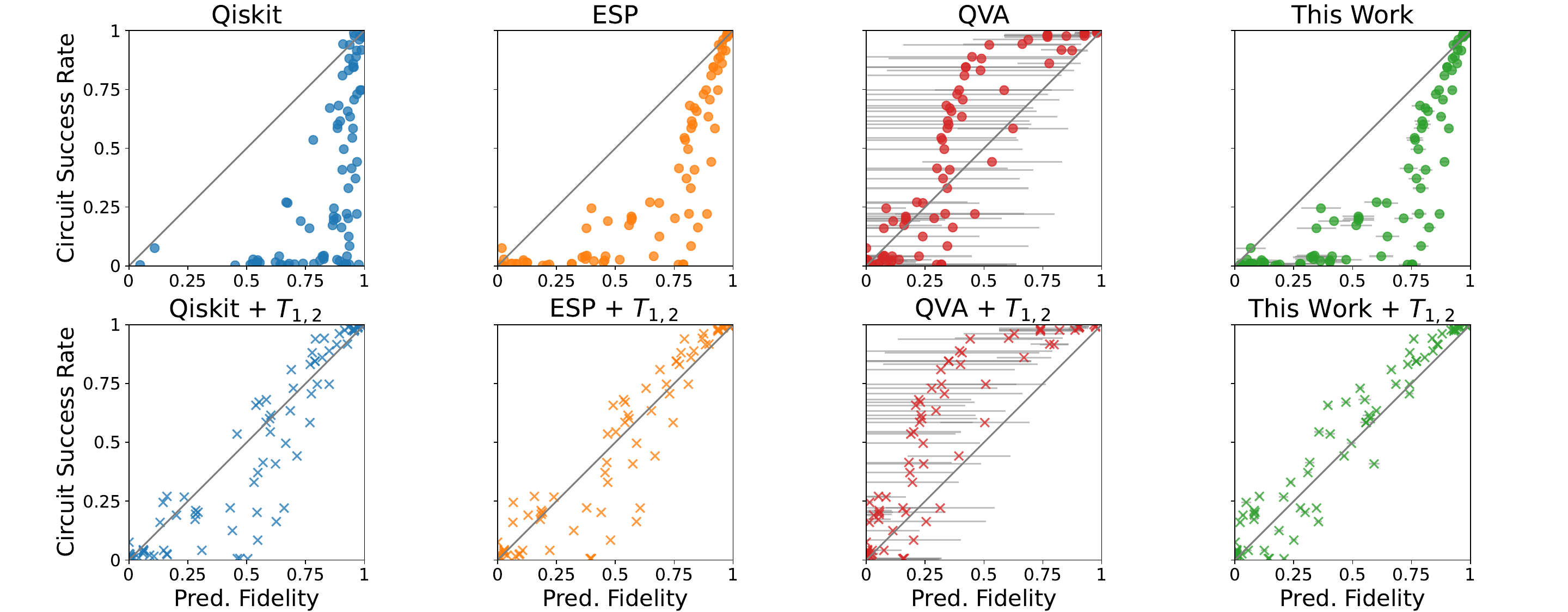}
    \caption{Comparison of predicted versus measured fidelities for all estimation models. Each scatter plot illustrates the predicted fidelity ($x$-axis) against the measured success rate ($y$-axis), with the diagonal line ($x=y$) indicating perfect predictions. For the QVA and the proposed model, horizontal lines indicate the range of predicted fidelities arising from the variability of their hyperparameters $w$ and $p_{ent}$, respectively.}
    \label{fig:correlation}
\end{figure}

Interestingly, QVA sometimes predicts a lower fidelity than the observed success rate despite not including coherence or crosstalk errors in its model. This discrepancy arises because QVA decreases the global fidelity twice for each two-qubit gate, doubling the reduction in fidelity compared to the actual gate error. When coherence errors are incorporated into the QVA model, the predicted fidelities decrease even further, as shown in the lower row of Figure \ref{fig:correlation}. Moreover, as mentioned before, the long horizontal lines for the QVA in Figure \ref{fig:correlation} highlight the significant influence that the hyperparameter for the cross-error ($w$) has on its predictions.

Surprisingly, we observe a significant discrepancy between Qiskit's fidelity estimation, obtained through full circuit simulation, and the measured fidelity on real hardware. Given that the simulation incorporates noise models inferred from the quantum device, we expected its predictions to be much closer to the actual fidelity. This deviation suggests potential inaccuracies in the extracted noise models.

Figure \ref{fig:fid_diff} presents the difference between each circuit's predicted and observed fidelity.

\begin{figure}
    \centering
    \includegraphics[width=1\linewidth]{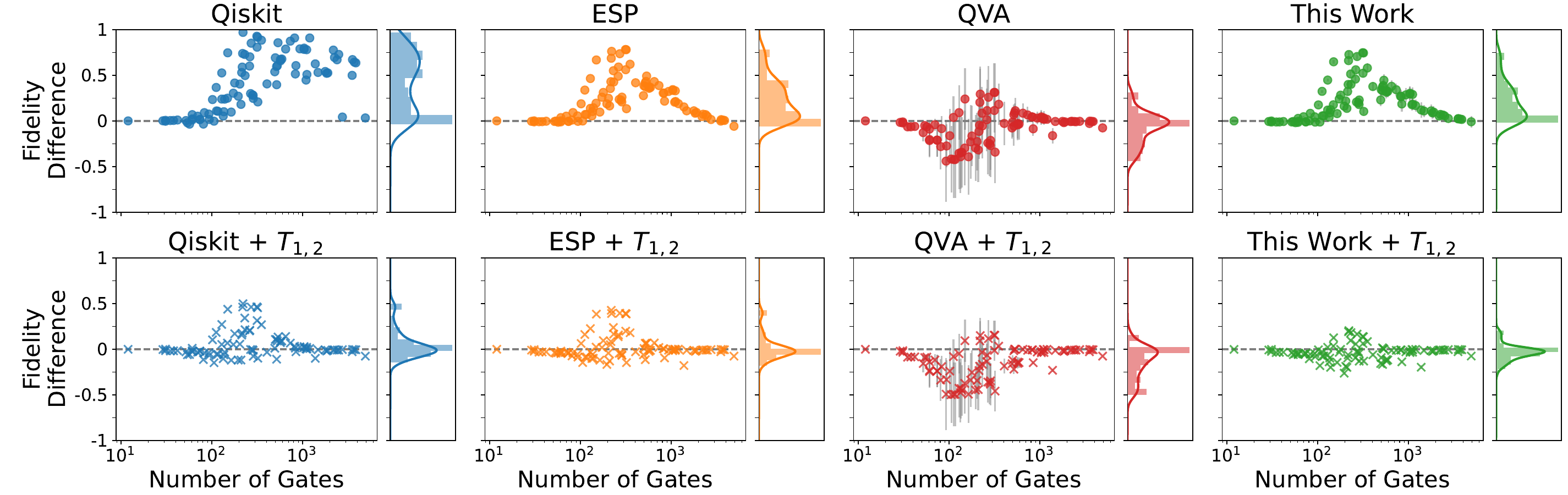}
    \caption{Comparison of the difference ($fid_{pred} - sr_{true}$) between predicted fidelity ($fid_{pred}$) and measured success rate ($sr_{true}$), plotted against the number of gates in each circuit for all fidelity estimation approaches.}
    \label{fig:fid_diff}
\end{figure}

Across all estimation techniques, we observe that, when only taking into account gate errors (upper row), for circuits with either very few or very many gates, the predictions tend to be relatively accurate. The explanation lies in the simplicity of estimating fidelity in these cases. For circuits with a very low number of gates, the fidelity is typically close to one, and fewer interactions among qubits reduce cross-talk errors, which are not considered by any of the evaluated models. Additionally, coherence noise effects are less significant due to the shorter execution times of such circuits. Conversely, the fidelity generally approaches zero for circuits with a high number of gates, causing the absolute difference between the predicted and real success rates to decrease.

To assess the accuracy of the proposed prediction models, we compute several regression metrics typically used to evaluate the performance of predictive models. These include the Mean Absolute Error (MAE) \cite{MAE}, which measures the average absolute difference between predicted and actual values; the Mean Squared Error (MSE) \cite{MSE}, which measures the average of the squares of the errors, that is, the average squared difference between the estimated values and the actual value; and the $R^2$ Score \cite{COLINCAMERON1997329}, which quantifies the proportion of variance in the observed data that is predictable from the model. Finally, the Pearson Correlation \cite{Pearson} coefficient is used to assess the linear relationship between the predicted and actual values, where a higher value indicates a stronger linear relationship. Additionally, the results with the cross-error parameter set to zero ($w=0$) for the QVA algorithm are also included since, in the analysis shown in \cite{10361567}, for most large circuits (over 100 two-qubit gates), the optimal $w$ was close to 0.

\begin{table}
    \centering
    \begin{tabular}{c|cccc}
         & \textbf{Mean Absolute} & \textbf{Mean Squared} & \textbf{$R^2$ Score} $\uparrow$ & \textbf{Pearson} \\
         & \textbf{Error (MAE) $\downarrow$} & \textbf{Error (MSE) $\downarrow$} &  & \textbf{Correlation} $\uparrow$ \\
         \hline
         \hline
        \textbf{Qiskit}                 & 0.422             & 0.273             & -0.831            & 0.606\\
        \textbf{Qiskit +$T_{1,2}$}     & 0.094             & 0.024             & 0.843             & 0.927\\
         \hline
        \textbf{ESP}                    & 0.229             & 0.098             & 0.341             & 0.827\\
        \textbf{ESP +$T_{1,2}$}        & 0.080             & 0.016             & 0.891             & 0.945\\
         \hline
        \textbf{QVA}                    & 0.134             & 0.034             & 0.769             & 0.903\\
        \textbf{QVA +$T_{1,2}$}        & 0.161             & 0.051             & 0.658             & 0.898\\
        \textbf{QVA$_{w=0}$}                    & 0.126             & 0.043             & 0.715             & 0.895\\
        \textbf{QVA$_{w=0}$ +$T_{1,2}$}        & 0.088             & 0.015             & 0.899             & 0.951\\
         \hline
        \textbf{This Work}              & 0.210             & 0.084             & 0.434             & 0.849\\
        \textbf{This Work +$T_{1,2}$} & \textbf{0.067}     & \textbf{0.008}    & \textbf{0.944}    & \textbf{0.976}\\
    \end{tabular}
    \caption{Regression metrics (Mean Absolute Error, Mean Squared Error, $R^2$ Score, and Pearson Correlation) for each fidelity estimation method: Qiskit simulation, ESP, QVA, and the proposed model, all with and without $T_{1,2}$. The best value for each metric is highlighted in bold, showcasing the relative performance of the approaches in predicting circuit fidelity with respect to the measured success rate on real hardware. Arrows indicate the preferred direction for each metric: for metrics where a higher value is better, the arrow points upwards ($\uparrow$), while for those where a lower value is more desirable, the arrow points downwards ($\downarrow$).}
    \label{tab:regression_metrics}
\end{table}

The results presented in Table \ref{tab:regression_metrics} highlight the performance of the proposed model (including $T_1$ and $T_2$), which consistently outperforms all other approaches across the selected metrics. It achieves the lowest MAE, improving between 15.63\% and 84.07\% over other techniques; the lowest MSE, with improvements ranging from 44.14\% to 96.92\%; the highest $R^2$ score, with gains of 4.96\% to 213.54\%; and the highest Pearson correlation, improving by 2.59\% to 61.05\%. The results including $T_{1,2}$ are a significant improvement across almost all models, demonstrating the importance of incorporating coherence noise into the fidelity estimation models. This suggests that including additional noise factors can significantly enhance the model's accuracy in predicting the fidelity of quantum circuits executed on a real device, making it a highly promising method for future quantum error modeling and prediction.

Since the QVA model already predicts a lower fidelity than the observed success rate even without accounting for $T_1$ and $T_2$ errors, incorporating coherence errors further worsens its predictions. This highlights the need for a stronger theoretical framework for more accurate fidelity estimation.

To ensure the robustness of our results, we repeated the experiments multiple times, conducting a total of six independent runs executed in different days. The proposed model consistently achieved the best performance across different calibration instances and quantum processors in all cases. Figure \ref{fig:day_pdf} provides a detailed comparison of the prediction differences observed across all estimation methods, as well as the obtained $R^2$ scores for each experiment.

The underlying code for this study and the data generated and analyzed are available in the GitHub repository, \url{https://github.com/escofetpau/Analytic-Model-of-Fidelity-under-Depolarizing-Noise}.

\section{Discussion}
\label{sec:discussion}
The performance differences between the considered fidelity estimation models reveal key insights into their underlying assumptions and applicability. The QVA model, despite showing competitive performance in regression metrics, predicts lower fidelities than those observed in real hardware for some circuits. This behaviour, combined with the exclusion of significant error sources such as coherence or crosstalk noise, suggests incomplete theoretical foundations. While the QVA model may provide reasonable predictions in some contexts, its tendency to deviate from hardware-observed outcomes raises questions about its reliability in diverse scenarios. It is, therefore, crucial not only to evaluate models based on final metrics but also to examine their assumptions and mechanisms, ensuring they align with the complexities of real devices.

\begin{figure}
    \centering
    \includegraphics[width=0.75\linewidth]{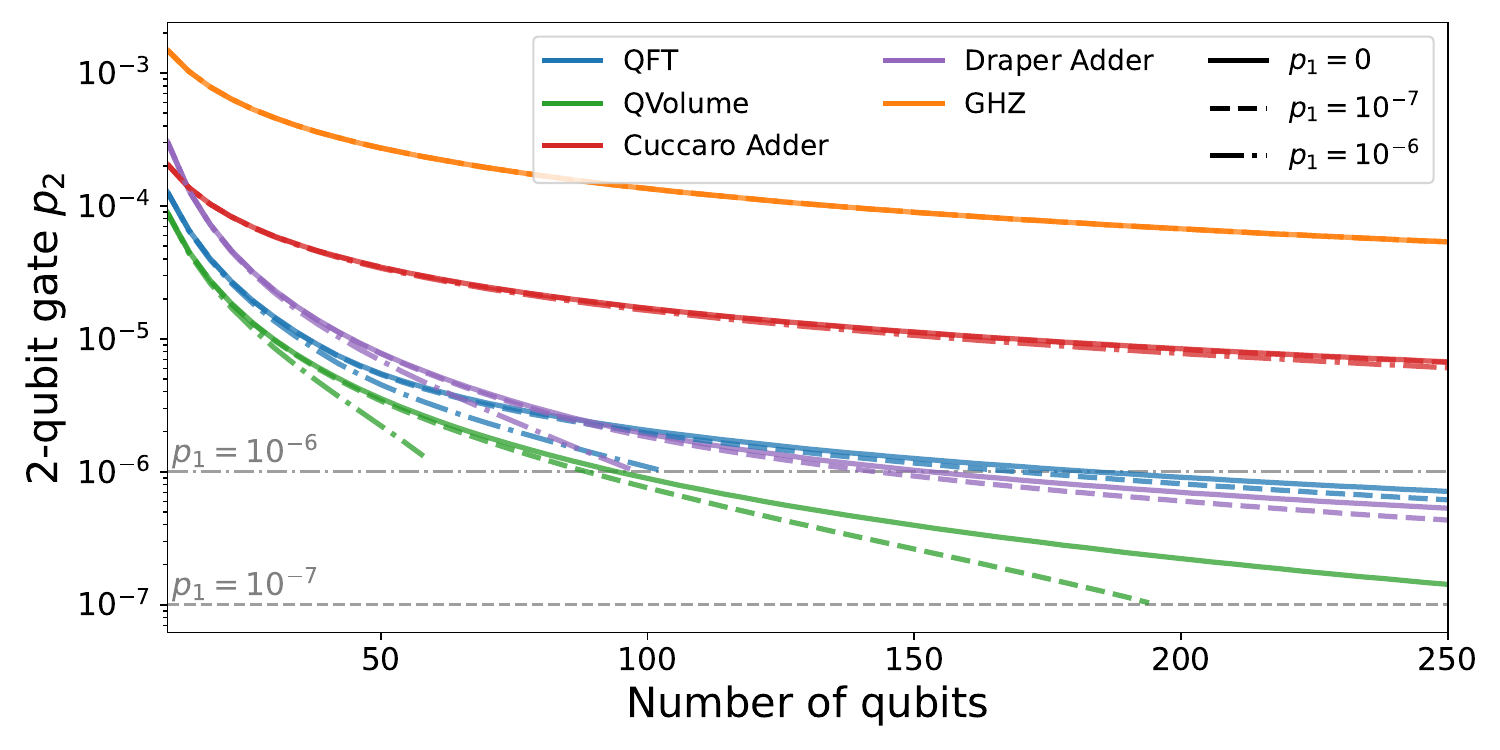}
    \caption{Required two-qubit depolarizing error rates ($p_2$) for achieving a circuit fidelity above 0.99 as the number of qubits increases, evaluated for various quantum circuits; QFT, Quantum Volume, Cuccaro Adder, Draper Adder, and GHZ State, obtained from Qiskit's circuit library \cite{javadiabhari2024quantumcomputingqiskit}. Each curve corresponds to a fixed single-qubit depolarizing error rate ($p_1$), with $p_1 = 0$, $10^{-7}$, and $10^{-6}$. The simulation is terminated when $p_2$ falls below $p_1$, indicating that achieving fidelity $> 0.99$ is not feasible for that number of qubits with the specified $p_1$. Coherence (\textit{i.e.}, $T_1$ and $T_2$) errors are excluded from this analysis, as the focus is solely on evaluating the impact of operational errors for design-oriented considerations.}
    \label{fig:link_design}
\end{figure}

Unlike simulations requiring extensive computational resources, the scalable depolarizing model proposed in this work can efficiently and accurately predict circuit fidelities, opening new avenues for exploration. For instance, it enables comparative studies of compilation strategies seeking to maximize computation fidelity, as showcased in existing literature \cite{10361567}. Additionally, the model facilitates the exploration of computational bounds for varying architectural parameters. By adjusting single- and two-qubit depolarizing factors ($p_1$ and $p_2$), the model can determine the conditions necessary to achieve high-fidelity computation. Figure \ref{fig:link_design} serves as an example. It illustrates, for various quantum circuits, the required $p_2$ values for achieving a circuit fidelity above 0.99 as the number of qubits increases for fixed $p_1$ values.

Such evaluations provide valuable insights for guiding quantum processor design. By establishing theoretical bounds on the required single- and two-qubit gate error rates, hardware developers can set performance targets before fabrication. This enables a more informed approach to optimizing hardware parameters, ensuring that the resulting quantum processors meet the fidelity requirements for practical computation.

Similarly, Figure \ref{fig:shor_computability} showcases the fidelity trends of Shor's algorithm \cite{shor_polynomial_1997} with 18 qubits for different combinations of $p_1$ and $p_2$. These results underscore the model's potential to guide architectural design and optimization, making it a powerful tool for advancing quantum computation.

These types of analysis are only possible thanks to the proposed estimation method's light scalability and high accuracy. If the model were inaccurate, it would lead to misleading predictions, potentially guiding hardware design and algorithmic decisions in the wrong direction. Additionally, without it being scalable, it would be unfeasible to thoroughly explore the full range of architectural parameters or extend the analysis to a higher number of qubits. By enabling precise and efficient fidelity estimations, the proposed approach allows for an in-depth evaluation of different quantum computing scenarios, supporting the development of more reliable and optimized quantum hardware and software.

\begin{figure}
    \centering
    \includegraphics[width=0.75\linewidth]{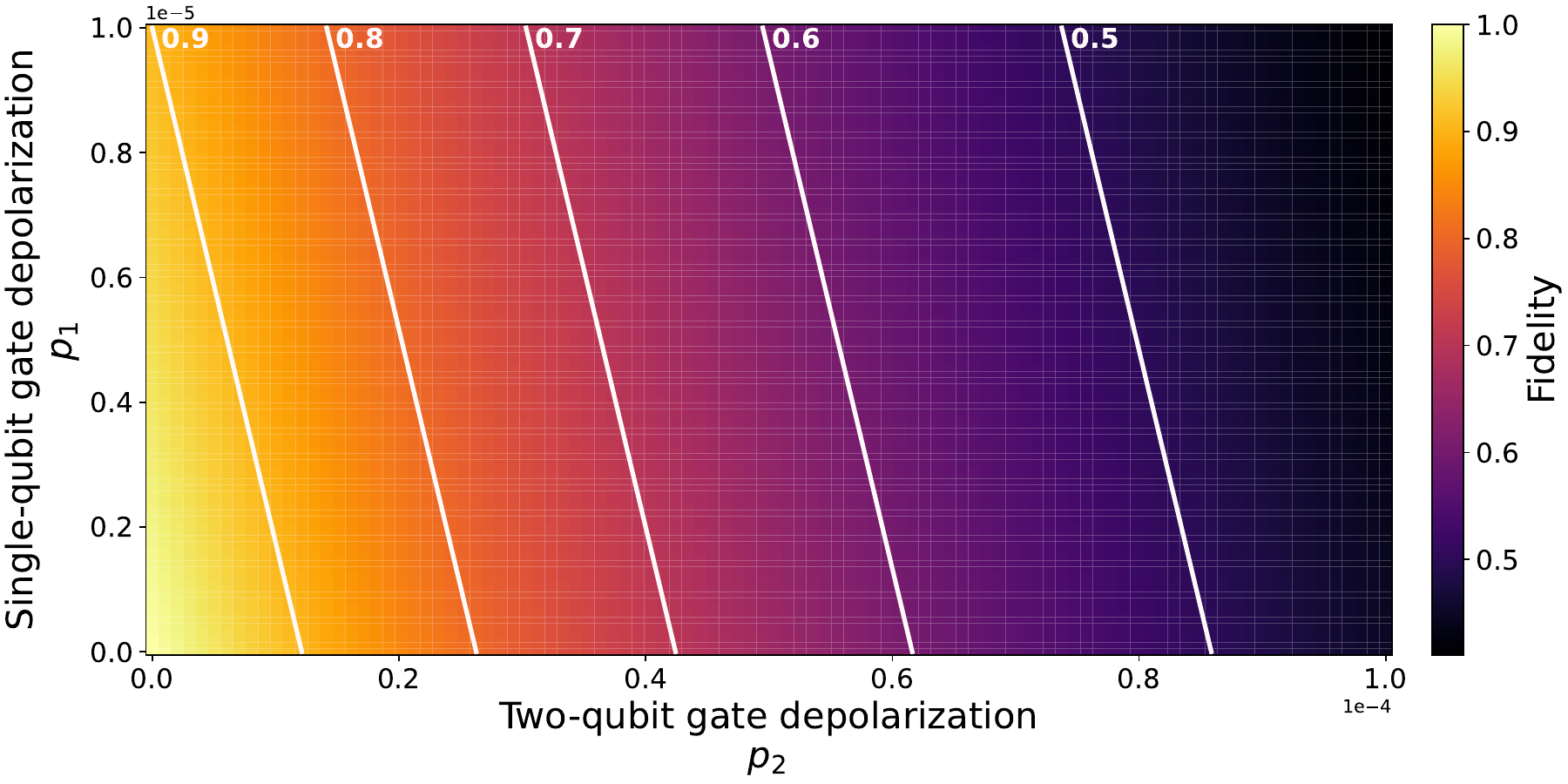}
    \caption{Estimated circuit fidelity for Shor’s algorithm \cite{shor_polynomial_1997}, containing 18 qubits, as a function of single-qubit ($p_1$) and two-qubit ($p_2$) gate depolarizing factors. The values of $p_1$ range from 0 to $10^{-5}$, while $p_2$ spans from 0 to $10^{-4}$. White contour lines indicate fidelity thresholds, illustrating the degradation of circuit fidelity as gate errors increase. These boundaries provide insight into the interplay between single- and two-qubit gate errors and their collective impact on algorithm performance. Coherence (\textit{i.e.}, $T_1$ and $T_2$) errors are excluded from this analysis, as the focus is solely on evaluating the impact of operational errors for design-oriented considerations.}
    \label{fig:shor_computability}
\end{figure}

\section{Conclusions}
\label{sec:conclusions}
Fidelity serves as the ultimate and most reliable metric for assessing the performance of quantum computations, directly reflecting the impact of noise and imperfections on quantum circuits. Accurate fidelity estimation is crucial for benchmarking quantum processors, optimizing compilation strategies, and guiding architectural decisions.

In this work, we have shown that state-of-the-art fidelity estimation techniques often fail to predict the performance of real-world quantum processors accurately. The discrepancies observed between predictions and actual hardware executions highlight the limitations of existing approaches, which either oversimplify noise models or lack scalability for large-scale quantum systems.

To address these challenges, we developed a theoretical framework that analytically characterizes the effects of depolarizing noise on quantum circuit fidelity. The proposed model provides an efficient and scalable approach to fidelity estimation, leveraging device calibration data to offer a more precise assessment of quantum computation reliability.

We validated our method against alternative estimation techniques and real hardware executions on IBM Q processors, consistently achieving the most accurate predictions across all tested scenarios. These results confirm the robustness of our approach in practical quantum computing environments.

Beyond its immediate application in fidelity estimation, our model's scalability enables a wide range of use cases, including hardware design optimization, performance evaluation of emerging quantum architectures, and feasibility analysis of large-scale computations. By providing a reliable and computationally efficient methodology, our work contributes to the advancement of quantum computing by offering a powerful tool for designing and assessing both quantum processors and quantum circuits.

\section*{Acknowledgements}
This work was supported by the European Commission (QUADRATURE: 101099697, WINC: 101042080), MCIN and NextGenerationEU (QCOMM-CAT), Project PCI2022-133004 funded by MCIN/AEI/10.13039/501100011033, by the Ministry for Digital Transformation and of Civil Service of the Spanish Government through the QUANTUM ENIA project call - Quantum Spain project, and by the European Union through the Recovery, Transformation and Resilience Plan - NextGenerationEU within the framework of the Digital Spain 2026 Agenda, and by the European Union NextGenerationEU/PRTR, and the ICREA Academia Award 2024. P.E. acknowledges support from an FPI-UPC grant funded by UPC and Banco Santander.

\appendix
\section{Robustness of Fidelity Estimation Across Multiple Hardware Executions}
Given the inherent variability in quantum hardware performance due to fluctuations in environmental conditions, calibration drift, and statistical noise in measurement processes, we conducted multiple independent experiments to validate the reliability of our fidelity estimation approach. Each experiment consisted of executing the same set of quantum circuits under comparable conditions, with device calibration data obtained prior to the execution to account for changes in hardware properties.

By repeating the experiments six times, we aimed to assess the consistency of our model’s predictions and its ability to generalize across different executions. The results confirm that our proposed approach systematically provides the most accurate fidelity estimates, demonstrating its robustness in real-world quantum computing scenarios. Figure \ref{fig:day_pdf} details the variations in fidelity predictions across experiments for all estimation methods, including coherence errors, and presents the corresponding $R^2$ scores to quantify the predictive accuracy in each case.

\begin{figure}[t]
    \centering
    \includegraphics[width=0.65\linewidth]{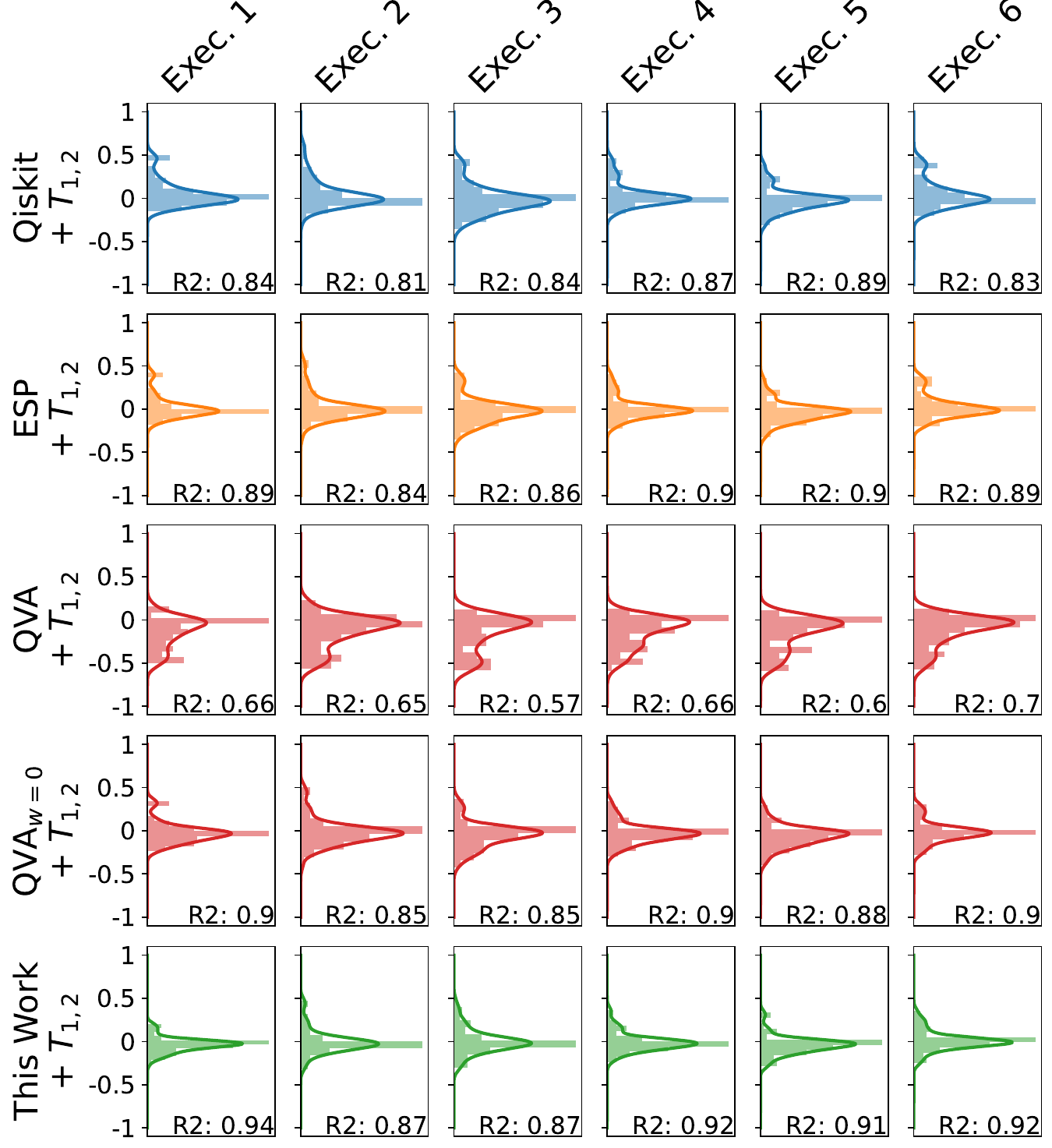}
    \caption{Comparison of the fidelity prediction accuracy for five different estimation models, each incorporating coherence errors ($T_1$ and $T_2$). The figure consists of five rows, one per model (including QVA with $w=0$), and six columns, each corresponding to an independent experiment conducted on different days using IBM Q Eagle processors (IBM Q Kyiv, IBM Q Sherbrooke, and IBM Q Brisbane). Each cell represents a specific experiment's fidelity error for each model. The $R^2$ metric is reported for each model, quantifying the correlation between predicted and measured fidelities. Across all experiments, the proposed model consistently achieves the highest $R^2$ values, ranging from $0.87$ to $0.94$, demonstrating superior predictive accuracy and robustness.}
    \label{fig:day_pdf}
\end{figure}

\section{Software Availability}
To ensure full reproducibility of our results, all code, quantum circuits (in QASM format), generated data, and analysis scripts used in this study are publicly available in our GitHub repository: \url{https://github.com/escofetpau/Analytic-Model-of-Fidelity-under-Depolarizing-Noise}.

This repository contains the source code for all fidelity estimation methods, including Qiskit Simulation, ESP, QVA, our proposed approach, and all circuits used in our experiments. Additionally, it provides six sets of experimental data collected from real hardware executions and the necessary scripts to regenerate all figures in the paper. By making our methodology accessible, we aim to facilitate independent validation and further research on fidelity estimation and quantum hardware performance analysis.

\section*{References}

\bibliographystyle{unsrt}
\bibliography{references}

\begin{thebibliography}{10}

\bibitem{preskill_quantum_2018}
John Preskill.
\newblock Quantum {Computing} in the {NISQ} era and beyond.
\newblock {\em Quantum}, 2:79, August 2018.

\bibitem{Proctor2022}
Timothy Proctor, Kenneth Rudinger, Kevin Young, Erik Nielsen, and Robin Blume-Kohout.
\newblock Measuring the capabilities of quantum computers.
\newblock {\em Nature Physics}, 18(1):75--79, Jan 2022.

\bibitem{khatri2020information}
Sumeet Khatri, Kunal Sharma, and Mark~M. Wilde.
\newblock Information-theoretic aspects of the generalized amplitude-damping channel.
\newblock {\em Phys. Rev. A}, 102:012401, Jul 2020.

\bibitem{Chapeau_Blondeau_2022}
François Chapeau-Blondeau.
\newblock Modeling and simulation of a quantum thermal noise on the qubit.
\newblock {\em Fluctuation and Noise Letters}, 21(06), September 2022.

\bibitem{ding_systematic_2020}
Yongshan Ding, Pranav Gokhale, Sophia~Fuhui Lin, Richard Rines, Thomas Propson, and Frederic~T. Chong.
\newblock Systematic crosstalk mitigation for superconducting qubits via frequency-aware compilation.
\newblock In {\em 2020 53rd Annual IEEE/ACM International Symposium on Microarchitecture (MICRO)}, pages 201--214, 2020.

\bibitem{Sarovar2020detectingcrosstalk}
Mohan Sarovar, Timothy Proctor, Kenneth Rudinger, Kevin Young, Erik Nielsen, and Robin Blume-Kohout.
\newblock Detecting crosstalk errors in quantum information processors.
\newblock {\em {Quantum}}, 4:321, September 2020.

\bibitem{Nachman2020}
Benjamin Nachman, Miroslav Urbanek, Wibe~A. de~Jong, and Christian~W. Bauer.
\newblock Unfolding quantum computer readout noise.
\newblock {\em npj Quantum Information}, 6(1):84, Sep 2020.

\bibitem{PhysRevA.52.R2493}
Peter~W. Shor.
\newblock Scheme for reducing decoherence in quantum computer memory.
\newblock {\em Phys. Rev. A}, 52:R2493--R2496, Oct 1995.

\bibitem{error_correction}
Joschka Roffe.
\newblock Quantum error correction: an introductory guide.
\newblock {\em Contemporary Physics}, 60(3):226--245, 2019.

\bibitem{Litinski2019gameofsurfacecodes}
Daniel Litinski.
\newblock A {G}ame of {S}urface {C}odes: {L}arge-{S}cale {Q}uantum {C}omputing with {L}attice {S}urgery.
\newblock {\em {Quantum}}, 3:128, March 2019.

\bibitem{10.1145/3123939.3123949}
Ali Javadi-Abhari, Pranav Gokhale, Adam Holmes, Diana Franklin, Kenneth~R. Brown, Margaret Martonosi, and Frederic~T. Chong.
\newblock Optimized surface code communication in superconducting quantum computers.
\newblock In {\em Proceedings of the 50th Annual IEEE/ACM International Symposium on Microarchitecture}, MICRO-50 '17, page 692–705, 2017.

\bibitem{Xu2024}
Qian Xu, J.~Pablo Bonilla~Ataides, Christopher~A. Pattison, Nithin Raveendran, Dolev Bluvstein, Jonathan Wurtz, Bane Vasi{\'{c}}, Mikhail~D. Lukin, Liang Jiang, and Hengyun Zhou.
\newblock Constant-overhead fault-tolerant quantum computation with reconfigurable atom arrays.
\newblock {\em Nature Physics}, 20(7):1084--1090, Jul 2024.

\bibitem{PRXQuantum.2.040101}
Nikolas~P. Breuckmann and Jens~Niklas Eberhardt.
\newblock Quantum low-density parity-check codes.
\newblock {\em PRX Quantum}, 2:040101, Oct 2021.

\bibitem{RevModPhys.95.045005}
Zhenyu Cai, Ryan Babbush, Simon~C. Benjamin, Suguru Endo, William~J. Huggins, Ying Li, Jarrod~R. McClean, and Thomas~E. O'Brien.
\newblock Quantum error mitigation.
\newblock {\em Rev. Mod. Phys.}, 95:045005, Dec 2023.

\bibitem{Takagi2022}
Ryuji Takagi, Suguru Endo, Shintaro Minagawa, and Mile Gu.
\newblock Fundamental limits of quantum error mitigation.
\newblock {\em npj Quantum Information}, 8(1):114, Sep 2022.

\bibitem{PRXQuantum.3.010345}
Yasunari Suzuki, Suguru Endo, Keisuke Fujii, and Yuuki Tokunaga.
\newblock Quantum error mitigation as a universal error reduction technique: Applications from the nisq to the fault-tolerant quantum computing eras.
\newblock {\em PRX Quantum}, 3:010345, Mar 2022.

\bibitem{Moll_2018}
Nikolaj Moll, Panagiotis Barkoutsos, Lev~S Bishop, Jerry~M Chow, Andrew Cross, Daniel~J Egger, Stefan Filipp, Andreas Fuhrer, Jay~M Gambetta, Marc Ganzhorn, Abhinav Kandala, Antonio Mezzacapo, Peter Müller, Walter Riess, Gian Salis, John Smolin, Ivano Tavernelli, and Kristan Temme.
\newblock Quantum optimization using variational algorithms on near-term quantum devices.
\newblock {\em Quantum Science and Technology}, 3(3):030503, jun 2018.

\bibitem{shor_polynomial_1997}
Peter~W. Shor.
\newblock Polynomial-time algorithms for prime factorization and discrete logarithms on a quantum computer.
\newblock {\em SIAM Journal on Computing}, 26(5):1484--1509, 1997.

\bibitem{feynman_1982_simulating}
Richard~P Feynman.
\newblock Simulating physics with computers.
\newblock {\em International Journal of Theoretical Physics}, 21:467--488, 06 1982.

\bibitem{4127220}
Oana Boncalo, Mihai Udrescu, Lucian Prodan, Mircea Vladutiu, and Alexandru Amaricai.
\newblock Using simulated fault injection for fault tolerance assessment of quantum circuits.
\newblock In {\em 40th Annual Simulation Symposium (ANSS'07)}, pages 213--220, 2007.

\bibitem{10.1145/3297858.3304007}
Swamit~S. Tannu and Moinuddin~K. Qureshi.
\newblock Not all qubits are created equal: A case for variability-aware policies for nisq-era quantum computers.
\newblock In {\em Proceedings of the Twenty-Fourth International Conference on Architectural Support for Programming Languages and Operating Systems}, ASPLOS '19, page 987–999, 2019.

\bibitem{10361567}
Fang Qi, Kaitlin~N. Smith, Travis LeCompte, Nian-feng Tzeng, Xu~Yuan, Frederic~T. Chong, and Lu~Peng.
\newblock Quantum vulnerability analysis to guide robust quantum computing system design.
\newblock {\em IEEE Transactions on Quantum Engineering}, 5:1--11, 2024.

\bibitem{9251243}
Ji~Liu and Huiyang Zhou.
\newblock Reliability modeling of nisq- era quantum computers.
\newblock In {\em 2020 IEEE International Symposium on Workload Characterization (IISWC)}, pages 94--105, 2020.

\bibitem{Vadali2023}
Avi Vadali, Rutuja Kshirsagar, Prasanth Shyamsundar, and Gabriel~N. Perdue.
\newblock Quantum circuit fidelity estimation using machine learning.
\newblock {\em Quantum Machine Intelligence}, 6(1):1, Dec 2023.

\bibitem{10.1145/3508352.3561118}
Hanrui Wang, Zhiding Liang, Jiaqi Gu, Zirui Li, Yongshan Ding, Weiwen Jiang, Yiyu Shi, David~Z. Pan, Frederic~T. Chong, and Song Han.
\newblock Torchquantum case study for robust quantum circuits.
\newblock In {\em Proceedings of the 41st IEEE/ACM International Conference on Computer-Aided Design}, ICCAD '22, 2022.

\bibitem{nielsen_chuang_2010}
Michael~A. Nielsen and Isaac~L. Chuang.
\newblock {\em Quantum Computation and Quantum Information: 10th Anniversary Edition}.
\newblock Cambridge University Press, 2010.

\bibitem{doi:10.1080/09500349414552171}
Richard Jozsa.
\newblock Fidelity for mixed quantum states.
\newblock {\em Journal of Modern Optics}, 41(12):2315--2323, 1994.

\bibitem{10.1145/3310273.3323053}
Jun Doi, Hitomi Takahashi, Rudy Raymond, Takashi Imamichi, and Hiroshi Horii.
\newblock Quantum computing simulator on a heterogenous hpc system.
\newblock In {\em Proceedings of the 16th ACM International Conference on Computing Frontiers}, CF '19, page 85–93, 2019.

\bibitem{10.5555/3135595.3135617}
Scott Aaronson and Lijie Chen.
\newblock Complexity-theoretic foundations of quantum supremacy experiments.
\newblock In {\em Proceedings of the 32nd Computational Complexity Conference}, CCC '17, Dagstuhl, DEU, 2017. Schloss Dagstuhl--Leibniz-Zentrum fuer Informatik.

\bibitem{9605307}
L.~Burgholzer, H.~Bauer, and R.~Wille.
\newblock Hybrid schrödinger-feynman simulation of quantum circuits with decision diagrams.
\newblock In {\em 2021 IEEE International Conference on Quantum Computing and Engineering (QCE)}, pages 199--206, 2021.

\bibitem{10044854}
Jaekyung Im and Seokhyeong Kang.
\newblock Graph partitioning approach for fast quantum circuit simulation.
\newblock In {\em Proceedings of the 28th Asia and South Pacific Design Automation Conference}, ASPDAC '23, page 690–695, 2023.

\bibitem{Boixo2018}
Sergio Boixo, Sergei~V. Isakov, Vadim~N. Smelyanskiy, Ryan Babbush, Nan Ding, Zhang Jiang, Michael~J. Bremner, John~M. Martinis, and Hartmut Neven.
\newblock Characterizing quantum supremacy in near-term devices.
\newblock {\em Nature Physics}, 14(6):595--600, Jun 2018.

\bibitem{biamonte2017tensor}
Jacob Biamonte and Ville Bergholm.
\newblock Tensor networks in a nutshell, 2017.

\bibitem{PhysRevA.77.012307}
E.~Knill, D.~Leibfried, R.~Reichle, J.~Britton, R.~B. Blakestad, J.~D. Jost, C.~Langer, R.~Ozeri, S.~Seidelin, and D.~J. Wineland.
\newblock Randomized benchmarking of quantum gates.
\newblock {\em Phys. Rev. A}, 77:012307, Jan 2008.

\bibitem{PhysRevA.85.042311}
Easwar Magesan, Jay~M. Gambetta, and Joseph Emerson.
\newblock Characterizing quantum gates via randomized benchmarking.
\newblock {\em Phys. Rev. A}, 85:042311, Apr 2012.

\bibitem{javadiabhari2024quantumcomputingqiskit}
Ali Javadi-Abhari, Matthew Treinish, Kevin Krsulich, Christopher~J. Wood, Jake Lishman, Julien Gacon, Simon Martiel, Paul~D. Nation, Lev~S. Bishop, Andrew~W. Cross, Blake~R. Johnson, and Jay~M. Gambetta.
\newblock Quantum computing with qiskit, 2024.

\bibitem{chow_2021_ibm}
Jerry Chow, Oliver Dial, and Jay Gambetta.
\newblock Ibm quantum breaks the 100‑qubit processor barrier, 11 2021.

\bibitem{PhysRevLett.119.180509}
Kristan Temme, Sergey Bravyi, and Jay~M. Gambetta.
\newblock Error mitigation for short-depth quantum circuits.
\newblock {\em Phys. Rev. Lett.}, 119:180509, Nov 2017.

\bibitem{10.1145/3674151}
James Ang, Gabriella Carini, Yanzhu Chen, Isaac Chuang, Michael Demarco, Sophia Economou, Alec Eickbusch, Andrei Faraon, Kai-Mei Fu, Steven Girvin, Michael Hatridge, Andrew Houck, Paul Hilaire, Kevin Krsulich, Ang Li, Chenxu Liu, Yuan Liu, Margaret Martonosi, David McKay, Jim Misewich, Mark Ritter, Robert Schoelkopf, Samuel Stein, Sara Sussman, Hong Tang, Wei Tang, Teague Tomesh, Norm Tubman, Chen Wang, Nathan Wiebe, Yongxin Yao, Dillon Yost, and Yiyu Zhou.
\newblock Arquin: Architectures for multinode superconducting quantum computers.
\newblock {\em ACM Transactions on Quantum Computing}, 5(3), September 2024.

\bibitem{guinn2023codesignedsuperconductingarchitecturelattice}
Charles Guinn, Samuel Stein, Esin Tureci, Guus Avis, Chenxu Liu, Stefan Krastanov, Andrew~A. Houck, and Ang Li.
\newblock Co-designed superconducting architecture for lattice surgery of surface codes with quantum interface routing card, 2023.

\bibitem{PRXQuantum.5.030326}
Lukas Postler, Friederike Butt, Ivan Pogorelov, Christian~D. Marciniak, Sascha Heu\ss{}en, Rainer Blatt, Philipp Schindler, Manuel Rispler, Markus M\"uller, and Thomas Monz.
\newblock Demonstration of fault-tolerant steane quantum error correction.
\newblock {\em PRX Quantum}, 5:030326, Aug 2024.

\bibitem{PRXQuantum.4.020345}
Zhenyu Cai, Adam Siegel, and Simon Benjamin.
\newblock Looped pipelines enabling effective 3d qubit lattices in a strictly 2d device.
\newblock {\em PRX Quantum}, 4:020345, Jun 2023.

\bibitem{t1_ibm}
{Qiskit Development Team}.
\newblock T1 characterization, 2021.

\bibitem{t2_ibm}
{Qiskit Development Team}.
\newblock T2 hahn characterization, 2021.

\bibitem{quetschlich2023mqtbench}
Nils Quetschlich, Lukas Burgholzer, and Robert Wille.
\newblock {MQT} {B}ench: {B}enchmarking {S}oftware and {D}esign {A}utomation {T}ools for {Q}uantum {C}omputing.
\newblock {\em {Quantum}}, 7:1062, July 2023.

\bibitem{10181370}
Ji~Liu, Max Bowman, Pranav Gokhale, Siddharth Dangwal, Jeffrey Larson, Frederic~T. Chong, and Paul~D. Hovland.
\newblock Qcontext: Context-aware decomposition for quantum gates.
\newblock In {\em 2023 IEEE International Symposium on Circuits and Systems (ISCAS)}, pages 1--5, 2023.

\bibitem{10.1145/3445814.3446718}
Casey Duckering, Jonathan~M. Baker, Andrew Litteken, and Frederic~T. Chong.
\newblock Orchestrated trios: compiling for efficient communication in quantum programs with 3-qubit gates.
\newblock In {\em Proceedings of the 26th ACM International Conference on Architectural Support for Programming Languages and Operating Systems}, ASPLOS '21, page 375–385, 2021.

\bibitem{qubit_allocation}
Marcos~Yukio Siraichi, Vin\'{\i}cius Fernandes~dos Santos, Caroline Collange, and Fernando Magno~Quintao Pereira.
\newblock Qubit allocation.
\newblock In {\em Proceedings of the 2018 International Symposium on Code Generation and Optimization}, CGO 2018, page 113–125, 2018.

\bibitem{li_2019_tackling}
Gushu Li, Yufei Ding, and Yuan Xie.
\newblock Tackling the qubit mapping problem for nisq-era quantum devices.
\newblock In {\em Proceedings of the Twenty-Fourth International Conference on Architectural Support for Programming Languages and Operating Systems}, ASPLOS '19, page 1001–1014, 2019.

\bibitem{10313754}
Ji~Liu, Ed~Younis, Mathias Weiden, Paul Hovland, John Kubiatowicz, and Costin Iancu.
\newblock Tackling the qubit mapping problem with permutation-aware synthesis.
\newblock In {\em 2023 IEEE International Conference on Quantum Computing and Engineering (QCE)}, volume~01, pages 745--756, 2023.

\bibitem{devoret2004superconductingqubitsshortreview}
M.~H. Devoret, A.~Wallraff, and J.~M. Martinis.
\newblock Superconducting qubits: A short review, 2004.

\bibitem{MAE}
C.~Willmott and K~Matsuura.
\newblock Advantages of the mean absolute error (mae) over the root mean square error (rmse) in assessing average model performance.
\newblock {\em Climate Research}, 30:79, 12 2005.

\bibitem{MSE}
Yadolah Dodge.
\newblock {\em Mean Squared Error}, pages 337--339.
\newblock Springer New York, New York, NY, 2008.

\bibitem{COLINCAMERON1997329}
A.~{Colin Cameron} and Frank~A.G. Windmeijer.
\newblock An r-squared measure of goodness of fit for some common nonlinear regression models.
\newblock {\em Journal of Econometrics}, 77(2):329--342, 1997.

\bibitem{Pearson}
Wilhelm Kirch, editor.
\newblock {\em Pearson's Correlation Coefficient}, pages 1090--1091.
\newblock Springer Netherlands, Dordrecht, 2008.

\end{thebibliography}

\end{document}